\documentclass[oneside,english]{amsart}
\usepackage[T1]{fontenc}
\usepackage[latin9]{inputenc}
\usepackage{amstext}
\usepackage{amsthm}
\usepackage{float}
\usepackage{graphicx}
\usepackage{xcolor}         % colors
\usepackage{multicol}
\usepackage{soul}
\usepackage{enumitem}

\makeatletter

\providecommand{\tabularnewline}{\\}
\floatstyle{ruled}
\newfloat{algorithm}{tbp}{loa}
\providecommand{\algorithmname}{Algorithm}
\floatname{algorithm}{\protect\algorithmname}

\providecommand{\tabularnewline}{\\}

\numberwithin{equation}{section}
\numberwithin{figure}{section}
\theoremstyle{plain}
\newtheorem{thm}{\protect\theoremname}
\theoremstyle{plain}
\newtheorem{lem}[thm]{\protect\lemmaname}
\theoremstyle{plain}
\newtheorem{prop}[thm]{\protect\propositionname}

\usepackage{babel}

\makeatother

\usepackage{babel}
\providecommand{\lemmaname}{Lemma}
\providecommand{\theoremname}{Theorem}
\providecommand{\propositionname}{Proposition}

\begin{document}

\title[Behavioral event detection and rate estimation for AV evaluation]{Behavioral event detection and rate estimation for autonomous vehicle evaluation}
\author{Maria A. Terres, Aiyou Chen, Rachel Zhou, Claire M. McLeod \\
\\
Waymo Data Science}
\date{\today}

\maketitle

\begin{abstract}
    Autonomous vehicles are continually increasing their presence on public roads. However, before any new autonomous driving software can be approved, it must first undergo a rigorous assessment of  driving quality.
 These quality evaluations typically focus on estimating the frequency of (undesirable) behavioral events. While rate estimation would be straight-forward with complete data, in the autonomous driving setting this estimation is greatly complicated by the fact that \textit{detecting} these events within large driving logs is a non-trivial task that often involves  human reviewers. In this paper we outline a \textit{streaming partial tiered event review} configuration that ensures both high recall and high precision on the events of interest. In addition, the framework allows for valid streaming estimates at any phase of the data collection process, even when labels are incomplete, for which we develop the maximum likelihood estimate and show it is unbiased. Constructing honest and effective confidence intervals (CI) for these rate estimates, particularly for rare safety-critical events, is a novel and challenging statistical problem due to the complexity of the data likelihood. We develop and compare several CI approximations, including a novel Gamma CI method that approximates the exact but intractable distribution with a weighted sum of independent Poisson random variables. There is a clear trade-off between statistical coverage and interval width across the different CI methods, and the extent of this trade-off varies depending on the specific application settings (e.g., rare vs. common events). In particular, we argue that our proposed CI method is the best-suited when estimating the rate of safety-critical events where guaranteed coverage of the true parameter value is a prerequisite to safely launching a new ADS on public roads. 
\end{abstract} 

\section{Introduction}
\subsection{Autonomous Driving Systems}
An ``autonomous driving system" (ADS) refers to a level 3, 4, or 5 system \cite{J3016}
that is capable of performing a dynamic driving task on a sustained basis.
It is comprised of both software and hardware components, and any vehicle equipped with this system may be referred to as an "Autonomous Vehicle" (AV). 

Development of safe and high quality ADSs is a very active area of research within both academia and industry \cite{Sun_2020_CVPR,yurtsever2020survey}. Waymo is a prominent player in this space, with its origins in 2009 as the Google Self-Driving Car Project. Today Waymo produces Level 4 ADSs which currently provide fully autonomous rides, i.e., with no human in the driver seat, to members of the public in Arizona and in California. Waymo also has parallel work in the autonomous trucking and delivery space. 

At Waymo, rigorous validation and evaluation of the ADS performance is a core tenet of the deployment philosophy. Prior to releasing AVs equipped with a new ADS onto public roads, we must have confidence that the safety and other performance metrics meet our stringent quality bar. Previous Waymo publications have already outlined the safety and readiness determinations  \cite{webb.etc:2020,schwall.etc:2020}. Complementing those publications, here we articulate some of the statistical and data science challenges that arise during the evaluation process. 

\subsection{ADS Performance Evaluation}
ADS performance evaluation typically involves three primary components:
\begin{enumerate}
    \item Definition of specific undesirable\footnote{There are many desirable behavioral events that are also of interest to track. For simplicity, in this paper we assume that these events can all be negated to produce a corresponding undesirable event. E.g., a desirable event may be ``sufficiently pullover to pick up passenger out of the travel lane," and the associated negative metrics would be things like ``pullover had X percent overlap with travel lane" or ``failed to pullover."} behavioral events, along with associated severity levels or other relevant annotations.
    \item Rate estimation for the behavioral events' relative frequencies on-road, typically in terms of the number of incidents per driving mile (IPM).
    \item Accommodation of product requirements from the stakeholders who use these data for decision-making.
\end{enumerate}
More details on each of these components are provided below.

The behavioral events of interest referenced in component (1) span a spectrum from critical safety situations such as collisions (very rare) to behavioral quirks such as unnecessary slowing (more common). While on its surface this aspect of performance evaluation appears straightforward, in practice rigorous definitions and delineation of edge cases can be a challenge. These challenges are left as out of scope for the current paper. 
Rather, we focus on data and statistical questions associated with the components (2) and (3), and how we are addressing these challenges at Waymo. 

There are several data and statistical questions that arise when producing the rate estimates referenced in component (2):
\begin{enumerate}[label=\alph*)]
    \item What data is best suited to produce these rate estimates? 
    Waymo has autonomously driven tens of millions of miles on public roads as well as  tens of billions of miles in simulation.
    The former have minimal observation error but may have sparse coverage of rare events; the latter may be subject to simulation artifacts but can more easily ensure coverage of rare events.
    
    \item Given the driving data, how can we reliably count these behavioral events? The ease of measurement for an individual class depends not only on the event's frequency (denser events are easier to measure empirically on smaller datasets), but also our ability to algorithmically identify relevant events in the driving data. Identification of behavioral events in a purely algorithmic fashion is difficult and usually requires human review, often of varying levels of expertise ("tiered event review").
    
    \item How can we quantify statistical uncertainty via reliable confidence intervals for the rate estimates? IPM is not a typical metric of interest within the technological industry outside of AVs, nor is it a quantity that receives broad attention from the academic literature. Moreover, the uncertainty introduced from imperfect event detection in (b) further complicates the statistical theory. Finally, these confidence intervals must accommodate a broad range of events, including rare event scenarios, for which traditional asymptotic methods can fail easily.
\end{enumerate}

These statistical challenges are further constrained by the product requirements referred to in component (3): 
\begin{itemize}
    \item In order to drive efficient decision-making at Waymo, these rate estimates must be available as soon as possible after mileage collection. As such, we need the ability to compute rate estimates even when not all candidate events have been reviewed. Early rate estimates may have greater uncertainty than the final estimates, but they should still be unbiased and provide robust confidence intervals.
    \item The most practically important event classes are very sparse (< 5-10 events), which means the most challenging statistical scenarios are also the most important.
    \end{itemize}

The rest of this paper outlines how we address 2(b) and 2(c) while accommodating these product constraints. In Section \ref{sec:triage} we introduce a \textit{streaming partial tiered event review} framework. Then, Sections \ref{sec:model}, \ref{sec:estimate}, \ref{sec:ci} outline the data generation model, estimation, and confidence interval construction implied by this framework. Section \ref{sec:sim} provides numerical studies to assess the performance of the proposed methods for both rare and common event cases. Finally, Section \ref{sec:discussion} concludes the paper with some discussion. Technical details are provided in the appendix, including a relation between Poisson, multivariate Hypergeometric and multinomial distributions, which may be of independent interest.

\section{The Role of Human Review in Event Detection}
\label{sec:triage}
\subsection{Behavioral Event Detection Overview}
As described above, computing the rate at which a behavioral event occurs will require us to first identify and count those events within our driving data. To facilitate this detection at scale, we rely on a set of sophisticated algorithms that can automatically identify events of interest. These classification problems are conceptually similar to many other data science applications in industry. However, unlike most other application settings, edge cases are incredibly important within the autonomous vehicles industry and cannot be swept aside or averaged out. These edge cases, for us, are frequently the most interesting and most important events.

Take, for example, safety-relevant behavioral events which occur very rarely. Errors in their classification will show up as a very small percentage of the algorithm's aggregate performance metrics, but lack of detection for these events can result in a very serious lack of information when deciding whether to put an ADS on public roads. This risk is compounded by the fact that estimating the recall loss by these algorithms can be very challenging and time-consuming since it requires finding sparse False Negatives in a very large collection of True Negatives.

To ensure confidence in our safety estimates at Waymo, we supplement these algorithmic solutions with a final stage of human review.  When tuning any algorithmic classifier, one can adjust the operation point to favor either precision or recall, with very few algorithms being able to achieve high performance in both simultaneously. Our solution allows us to tune the algorithms for very high recall, and then lean on a final layer of human reviewers to fix any gaps in precision. Combined, the algorithm and human reviewers represent a multi-stage strategy that is both scalable and provides trustworthy rate estimates.

\subsection{Review of Candidate Events}
\label{subsec:streaming_triage}
Let $\mathcal{E}_0$ refer to the corpus of candidate events identified by the algorithms. For simplicity, we assume the algorithm feeding events to the human reviewers has perfect recall, so the human reviewers only need to address precision. Then, for any candidate event in $\mathcal{E}_0$, we ultimately want to confirm whether the event was a True Positive or False Positive detection of the behavioral event of interest.\footnote{Note that detailed instructions and rigorous training are often needed to ensure that human reviewers provide high quality labels, particularly for the most nuanced behavioral events. For the purposes of this paper, we assume the human labels are perfect and, as such, recall from this system is perfect.}

A well-optimized event review configuration should maintain both high precision and high recall for True Positive events. If the precision is low, then the rate estimates will be too high; if the recall is low, then the rate estimates will be too low. With this as our motivation, we have developed an event review configuration at Waymo referred to as \textit{streaming partial tiered} event review. The core concepts motivating this configuration are outlined below. 

Under a \textit{complete} event review configuration, all candidate events would be reviewed. However, for many applications this can be inefficient and labor intensive. Instead, reviewing a subset of the candidate events, i.e., \textit{partial} event review, is often sufficient to estimate event rates with acceptable levels of uncertainty. Beyond efficiency, supporting estimates based on partial event review also ensures that we can fulfill the product requirements by providing trustworthy estimates at any stage of the event review process (e.g., checking results after only 10\% of events have received review), and that the uncertainty in those estimates will continually be reduced as more events are reviewed.\footnote{In addition to ongoing event review, the on-road ADS logs and simulated autonomous driving provide a continuous stream of candidate events that are being fed into the event review pipeline. As a result, the corpus of candidate events is constantly changing. For simplicity, this is not addressed here, and we instead assume the full corpus of candidate events is available and unchanging prior to the event review initiation.}

In its simplest form, partial event review utilizes uniform random sampling to select a subset of candidate events for labeling, but this is unlikely to be the most efficient strategy for most use-cases. To reduce sampling error, \textit{stratified random sampling} can instead be used; this involves partitioning candidate events into mutually exclusive, cumulatively exhaustive strata based on pertinent event metadata (e.g. road characteristics, algorithmic model scores). The sampling probability may be identical across strata (i.e. proportional allocation), or may differ across strata (e.g. Neyman allocation or alternative allocation schemes). In some sections of this paper we simplify to assume a single stratum, without loss of generality, while in other sections we outline the statistical implications of this stratification in the data. 

For operational efficiency, an event review configuration may utilize several tiers of human reviewers with increasing levels of expertise in the higher tiers. This \textit{tiered} configuration is arranged such that each tier improves the precision of the candidate event corpus while maintaining recall. For example, in a three tiered process, the first tier would exclude clear false positive candidate events while escalating all other events to the second tier. The second tier would exclude additional false positive candidate events, while escalating the remainder to the third tier. The third tier would then make the final determination of false positive vs. true positive status. Each tier may itself consist of evaluation by a single reviewer or multiple reviewers who reach consensus.

Operationally, partial tiered event review can be implemented as either a \textit{batched} or \textit{streaming} process. Under a batched process, any candidate event that is reviewed in Tier 1 will complete its journey through the system and be confirmed as either a False Positive or make it to the final tier and be confirmed as a True Positive. In this setting, the rate estimates  do not need to account for the details of how candidate events are escalated between tiers; the rate estimates are ultimately identical between a partial tiered setting and a partial \textit{un}tiered setting. However, at Waymo it is often the case that we need to produce unbiased rate estimates while many events are still only partway through the event review pipeline. E.g., a candidate event may have already been reviewed by Tier 1 and escalated, but hasn't received further review yet. This is considered to be a streaming process where rate estimates are made based on the data available at any snapshot in time.

\section{Statistical Formulation for Partial Tiered Event Review}
\label{sec:model}
Given a collection of driving data representing $m$ autonomous driving miles, we are interested in providing an estimate and confidence interval for the rate of occurrence ($\theta$) for the behavioral event of interest. Let $T$ be the number of tiers in the event review configuration and $x_T$ be the total number of True Positive behavioral events in the driving data. We assume that $x_T$ follows a Poisson distribution with mean $\theta m$. Note that $x_T$ is not directly observed; instead, a corpus of candidate events, $\mathcal{E}_0$ are sent through the streaming tiered partial event review process described in Section \ref{subsec:streaming_triage}, resulting in a collection of partial event review results for each tier.\footnote{For simplicity, assume that the algorithmic approach to identify the corpus of candidate events has perfect recall. With this assumption, we don't need to consider its role in the following sections.}  

We first discuss the statistical framing for a single stratum, and then extend this to accommodate multiple strata in a straightforward manner.

% TODO: Make sure the alpha notation is defined later.

\subsection{One stratum} 
\label{subsec:mode_one}

Let $e_{0} = ||\mathcal{E}_0||$ be the number of candidate events in the full corpus. 

\subsubsection{Partial Event Review}

Due to the partial event review configuration, only a subset of the $e_{0}$ candidate events will be randomly sampled and receive review in Tier 1; denoted by $n_{1}\leq e_{0}$. The reviewers will improve the precision of the event pool by clearly labeling the most obvious False Positives (removing them from further review) and escalating any ambiguous events to Tier 2. Let $e_{1} \leq n_{1}$ be the number of events that Tier 1 escalates for review by Tier 2. This process proceeds at each tier with $n_{t} \leq e_{t-1}$ denoting the number of events reviewed by Tier $t$, and $e_{t} \leq n_{t}$ denoting the number of events escalated by Tier $t$ up to Tier $t+1$. Finally, $e_{T}$ will be the number of candidate events formally confirmed as True Positives by Tier $T$.

\begin{figure}
\centering{}
\includegraphics[width=1\textwidth]{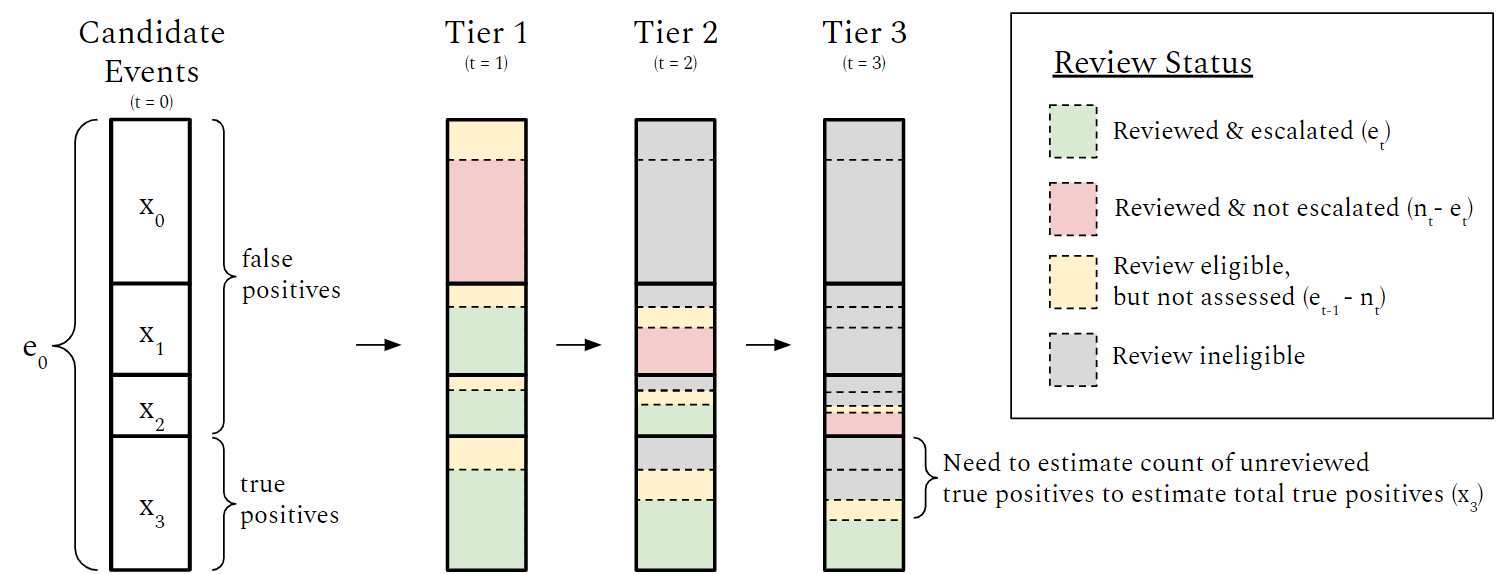} \caption{\label{fig:one_strata_example} Illustration of the sampling and escalation of events in the partial event review process from Tier 1 to Tier 3 (assuming $T=3$ for a single stratum), where $x_3$ is a latent variable representing the number of True Positive events and $x_t$ for $t\in \{0, 1, 2\}$ represents the number of candidate events which if reviewed would be rejected as False Positive by Tier-$(t+1)$ (formally defined in Algorithm \ref{alg:data-model-one-stratum}).}
\end{figure}

Figure \ref{fig:one_strata_example}
provides an illustration of the sampling and escalation process in a single
stratum under a partial event review configuration with 3 Tiers, and the detailed data generation process is described next. 

\subsubsection{Data Generation Model}

Recall that our goal is to estimate the True Positive rate, $\theta$, for the behavioral event of interest. Under complete event review, the total confirmed True Positive count $e_T$ would be equal to the total True Positive count in $\mathcal{E}_0$, and the estimate of $\theta$ would be a simple function of $e_T$. 
However, because our event review is \textit{partial}, we know that $e_T \le x_T$; i.e., there are additional True Positives within the event review pipeline that have not yet been reviewed by the top tier. We have developed an estimation scheme that includes the contribution from these as-of-yet unreviewed True Positives. Towards this end, we propose a data generation model for the \textit{full} set of candidate events in the event review pipeline, including both True Positives \textit{and} False Positives.

%%%%%%%%%%%%%%%%%%%%%%%%%%%%%%%%%%%%%%%%
% The urn model version (11/20/2022)
%%%%%%%%%%%%%%%%%%%%%%%%%%%%%%%%%%%%%%%%
To introduce the data generation model, let's consider the potential outcome of each candidate event in $\mathcal{E}_0$ under a \textit{complete} event review configuration. Under the complete event review configuration, a candidate event would first be reviewed by Tier-1, then it would be either rejected as False Positive by Tier-1, or escalated by Tier-1; were the candidate escalated by Tier-1, it would be reviewed by Tier-2, and then it would be either rejected by Tier-2 as False Positive, or escalated by Tier-2, and so on. That is, the event review process continues until the candidate event would be either rejected as False Positive by Tier-$t$ (labeled as FP-$t$) for some $t\in \{1,\cdots, T\}$ or would not be rejected by any Tier (labeled as TP). Whenever a candidate event is rejected as False Positive, it is removed from the event review pipeline for further consideration.

Let $x_T$ be the number of candidate events in $\mathcal{E}_0$ which would be labeled as TP, and let $x_{t-1}$ be the number of candidate events  in $\mathcal{E}_0$ which would be labeled as FP-$t$ for $t\in \{1, \cdots, T\}$,  under the complete event review configuration. 
It is natural to assume that the $x_t$ are independent counts and that they each follow some Poisson process, i.e. for $m$ autonomous miles,
\begin{equation}
x_{t}\sim Poisson(\lambda_{t}m)\text{ for }t=0,\cdots T.\label{eq:model}
\end{equation}
Note that $\sum_{s=t}^Tx_s$ has a simple interpretation: it is the total number of candidate events in $\mathcal{E}_0$ which if fully reviewed would not be rejected by Tier-$t$.
Again, in the \textit{complete} event review case these $x_{s}$ would be observed, but in the \textit{partial} event review configuration these are instead treated as latent variables; only their total is observed, i.e.
$$e_0 = \sum_{t=0}^T x_t.$$
%With a single stratum, the parameter of interest is $\theta =\lambda_{T}.$

Now given the set of candidate events $\mathcal{E}_0$, we draw a random sample of size $n_1$ without replacement for Tier-1 review, among which the subset of candidate events with labels other than FP-$1$ --- that is, not rejected by Tier-1 --- will be denoted as $\mathcal{E}_1$ and will be escalated for Tier-2 review. 
The process continues as follows: for $t=2, \cdots, T$, we draw a random sample of size $n_t$ from $\mathcal{E}_{t-1}$ without replacement for Tier-$t$ review, among which the subset of candidate events with labels other than FP-$t$ --- that is, not rejected by Tier-$t$ --- will be denoted as $\mathcal{E}_t$ and will be escalated for Tier-$(t+1)$ review. 
Let $e_t$ be the number of candidate events escalated by Tier-$t$:
$$e_t = ||\mathcal{E}_t||.$$
Obviously, $n_t - e_t$ is the number of candidate events which are reviewed but rejected as False Positives by Tier-$t$. 
It is also important to note that any candidate event in $\mathcal{E}_t$ will be either TP or FP-$s$ for some $s\in \{t+1, \cdots, T\}$.

\begin{algorithm}[!t]
\caption{Data generation model for partial event review with single stratum}
\label{alg:data-model-one-stratum}

Configuration: $m$ miles and $T$ tiers.

Parameters: $\{\lambda_{t}:0\leq t\leq T\}$ and $\{\pi_{t}:1\leq t\leq T\}$.

\begin{enumerate}
\item Draw
\begin{align*}
x_{t} &\sim Poisson(m \lambda_{t}), \forall t \in \{0,\cdots, T\},
\end{align*}
where $x_T$ is the number of (latent) True Positive events (labeled as TP), and for $t=1,\cdots, T$, 
\begin{itemize}
\item $x_{t-1}$ is the number of (latent) candidate events which are labeled as FP-$t$ --- that is, they would be escalated by Tier-$(t-1)$ but rejected as False Positive by Tier-$t$.
\end{itemize}
Then the total number of candidate events in the event review pipeline is
\begin{align*}
e_0 & = \sum_t x_{t}.
\end{align*}
Let $\mathcal{E}_0$ consist of all $e_0$ candidate events.

\item For tier $t=1,\cdots,T$,
\begin{itemize}
\item If $e_{t-1}>0$:
\begin{itemize}
  \item Let $n_t$ be the number of candidate events which are drawn randomly from $\mathcal{E}_{t-1}$ without replacement and are reviewed by Tier-$t$:
    \begin{align*}
      n_t &\sim \max(1,Binomial(e_{t-1},\pi_{t})).
    \end{align*}
  \item Let $\mathcal{E}_t$ be the subset of the $n_t$ candidate events with labels other than FP-$t$ --- that is, they are escalated by Tier-$t$.
  \item Set $e_{t}=||\mathcal{E}_t||$.
\end{itemize}
\item Else: 
\begin{itemize}
    \item Set $n_s=0$ and $e_s=0$ for $s=t, \cdots, T$.
    \item Break (i.e. early termination).
\end{itemize}
\end{itemize}
\item Output: $\{e_{t}:0\leq t\leq T\}$ and $\{n_{t}:1\leq t\leq T\}$ as  the observed data.
\end{enumerate}
\end{algorithm}

Of course, if $\mathcal{E}_t$ is empty for some $t<T$, i.e. $e_t=0$ (no candidate events are escalated by Tier-$t$), then the partial event review process terminates at Tier-$t$.

It is convenient to model the sample size $n_t$ by a Binomial process. Note that if $e_{t-1} > 0$ but $n_t=0$, i.e. there are $e_{t-1}$ candidate events escalated for Tier-$t$ event review, but none is reviewed by Tier-$t$, then $\theta$ is not identifiable. To make the estimation problem well defined, whenever $e_{t-1}>0$, we need $n_t \geq 1$, and to be precise, we use
\begin{align}
n_{t}\sim \max(1, Binomial(e_{t-1},\pi_{t})),\label{eq:sample}
\end{align}
where $Binomial(e_{t-1}, \pi_t)$ stands for the random variable from a Binomial distribution with parameters $e_{t-1}$ and $\pi_t \in (0, 1]$.

This data generation model is summarized as Algorithm \ref{alg:data-model-one-stratum}, which takes $\{\lambda_t: 0 \leq t \leq T\}$ (associated with the latent variables $x_t$) and $\{\pi_t: 1 \leq t \leq T\}$ (associated with tier level sampling) as the input parameters, and outputs the tier level sample sizes ($n_t$) and the numbers of tier level escalations ($e_t$) as the observed data.

\subsection{Multi-strata} Now suppose we sample from the candidate events via a stratified sampling scheme with $H$ comprehensive and mutually exclusive strata. We can generalize all the notations in Section \ref{subsec:mode_one} in a straightforward way by an additional subscript $h \in \{1\cdots H\}$ indicating the stratum. The detailed data generation model is provided as Algorithm \ref{alg:data-model}, which calls Algorithm \ref{alg:data-model-one-stratum} for each stratum $h$ by taking $\{\lambda_{ht}: 0 \leq t \leq T\}$ (associated with latent variables) and $\{\pi_{ht}: 1 \leq t \leq T\}$ (associated with tier level sampling) as the input parameters, and outputing the tier level sample sizes ($n_{ht}$) and the numbers of tier level escalations ($e_{ht}$) as the observed data. The parameter of interest can be written as 
\begin{align}
\theta = \sum_{h=1}^T\lambda_{hT}. \label{eq:theta}
\end{align}

\begin{algorithm}
\caption{Data generation model for partial event review with multiple strata}
\label{alg:data-model}
Configuration: $m$ miles, $H$ strata, and $T$ tiers.

Parameters: $\{\lambda_{ht}:0\leq t\leq T\}$ and $\{\pi_{ht}:1\leq t\leq T\}$
for $h=1,\hdots,H$. 
%$\theta=\sum_{h=1}^T\lambda_{hT}$.

\begin{itemize}
\item For each $h\in \{1,\cdots, H\}$, 
\begin{itemize}
\item use $\{\lambda_{ht}: 0\leq t\leq T\}$ and $\{\pi_{ht}: 1\leq t\leq T\}$ as the input parameters and apply Algorithm \ref{alg:data-model-one-stratum} to generate the output for stratum-$h$:
$\{e_{ht}:0\leq t\leq T\}$ and $\{n_{ht}: 1\leq t\leq T\}$.
\end{itemize}

\item Output: $\{e_{ht}\}$ and $\{n_{ht}\}$ as the observed data for all strata.
\end{itemize}
\end{algorithm}

\section{Point estimation}
\label{sec:estimate}
%An ideal estimate of $\theta = \sum_h\lambda_{hT}$ would be the maximum likelihood estimate (MLE) since it has various optimality properties 
To estimate $\theta$, it is desirable to derive the maximum likelihood estimate (MLE) due to its various optimality properties (\cite{bickel2015mathematical,yatracos1998small,lehmann2006theory}) with the hope that it is also unbiased.
Instead of deriving the MLE directly, we first derive the point estimate based on a heuristic argument, then show that it is indeed the MLE. This is due to the complexity of the likelihood function under the partial event review configuration, with more details provided in the Appendix. The estimate is also unbiased.

Since the data across strata are independent, the rates associated with different strata can be estimated independently.
%\textbf{[Rachel to fix and then unbold] 
Without loss of generality, we may consider a single stratum, using the same notation as in Section \ref{subsec:mode_one}.
%since, in practice, there are often some events confirmed as True Positives for the outcomes we care about.

A natural estimate of $m\lambda_{T}$
is the total number of true positives within $e_{0}$ (all candidate
events in the stratum). This quantity is unfortunately not observable under partial event review and must be estimated instead. We start the argument by assuming $e_{T}>0$. 
%Note that $e_{h0}\geq n_{h1}\geq\cdots\geq e_{hT}$. 

First, consider the total number of True Positive events present in the final tier, Tier $T$. Recall that $e_{(T-1)}$ candidate events will be escalated to tier $T$, a random sample of size $n_{T}$ will be reviewed, and $e_{T}$ of those reviewed events will be confirmed as True Positives. Then, one can estimate the total expected number of true positives that were present in Tier $T-1$  as below by up-weighting the observed true positives by the inverse of the sampling weight:
\begin{align*}
e_{T}\left(\frac{n_{T}}{e_{(T-1)}}\right)^{-1}.
\end{align*}
By the same logic, one can estimate the expected number of true positives that were present in Tier $T-2$
to be 
\begin{align*}
e_{T}\left(\frac{n_{T}}{e_{(T-1)}}\right)^{-1}\left(\frac{n_{(T-1)}}{e_{(T-2)}}\right)^{-1} & .
\end{align*}
Iterating through each of the tiers leads to the estimate for the total number of true positives in all $e_0$ candidate events in the stratum:
\begin{align*}
e_{T}\left(\frac{n_{T}}{e_{(T-1)}}\right)^{-1}\left(\frac{n_{(T-1)}}{e_{(T-2)}}\right)^{-1}\cdots\left(\frac{n_{1}}{e_{0}}\right)^{-1} & .
\end{align*}
Normalizing the above by $m$ after simplification leads to the estimate of
$\lambda_{T}$: 
\begin{align}
\hat{\lambda}_{T} & =\frac{1}{m}\frac{\prod_{t=0}^{T}e_{t}}{\prod_{t=1}^{T}n_{t}}.\label{eq:hat-lambda-hT}
\end{align}

Note, an edge case occurs when $e_{T}=0$, or even $e_{t}=0$ for some $t<T$;
that is, when no events are escalated to Tier $(t+1)$ for some $t$. In these cases the partial event review
process would terminate early, and the natural estimate is $\hat{\lambda}_{T}=0$.

For $t \in \{0, \cdots, T\}$, let
$$\Lambda_t = \sum_{s=t}^T \lambda_s.$$
That is, $\Lambda_t$ is the overall Poisson rate for candidate events in $\mathcal{E}_0$ which if fully reviewed would not be rejected by Tier-$t$ for $t\in\{1, \cdots, T\}$. Note $\Lambda_0\geq \cdots \geq \Lambda_T$.
%then $$\lambda_T \equiv \Lambda_T$$
%and  for $t=0, \cdots, T-1$,
%$$\lambda_t = \Lambda_t - \Lambda_{t+1}.$$

Similar to the heuristic derivation of \eqref{eq:hat-lambda-hT}, one can estimate $\Lambda_t$ as follows:
$$\hat{\Lambda}_0 = \frac{e_0}{m}$$ 
and for $t\geq 1$,
\begin{align}
    \hat{\Lambda}_{t}=\frac{1}{m}\frac{\prod_{s=0}^{t}e_{s}}{\prod_{s=1}^{t}n_{s}}. \label{eq:Lambda-hat}
\end{align}
Note that $\hat{\Lambda}_t$ is well defined whenever $e_{t-1}\geq 1$, which is guaranteed by the sampling model \eqref{eq:sample}. If $\hat{\Lambda}_t=0$, which may occur due to $e_t=0$ (i.e. early termination if $t<T$, see Algorithm \ref{alg:data-model-one-stratum}), then $\hat{\Lambda}_s = 0$ for $s\geq t$ due to monotonicity. Therefore, the estimates are always well-defined. Note that $\hat{\Lambda}_T \equiv \hat{\lambda}_T$. In fact, the estimates are not only MLE but unbiased.

\begin{thm}
\label{thm:MLE}Under the Poisson model \eqref{eq:model} and the sampling model \eqref{eq:sample},
\begin{enumerate}
\item $\{\hat{\Lambda}_t: 0\leq t\leq T\}$ as defined above %by (\ref{eq:hat-lambda-hT})
is the MLE of $\{\Lambda_t: 0\leq t\leq T\}$.
\item The point estimates are unbiased, i.e. $E(\hat{\Lambda}_t)=\Lambda_t$, $\forall t\in\{0, \cdots, T\}$. 
\item The MLE for $\{\lambda_{t}\}$ for $t=0,\cdots,T-1$ is given by
\begin{align}
\hat{\lambda}_{t} & = \hat{\Lambda}_{t}-\hat{\Lambda}_{(t+1)}.\label{eq:lambda-hat}
\end{align}
%where $\hat{\Lambda}_{0}=e_{0}$ and $\hat{\Lambda}_{t}=\frac{\prod_{s=0}^{t}e_{t}}{\prod_{s=1}^{t}n_{t}}$ for $t=1,\cdots,T$.
\end{enumerate}
\end{thm}

The derivation of the MLE for $\{\Lambda_{t}\}$ (and thus $\{\lambda_t\}$) relies
on a novel application of the Expectation-Maximization algorithm \cite{dempster1977maximum}.
The complete proof is provided in the Appendix.

For multi-strata, since the data across strata are assumed to follow independent random processes, we can apply the MLE estimator described in Theorem  \ref{thm:MLE} to each stratum $h$ and obtain the estimate $\hat{\lambda}_{ht}$ for $t=0,\cdots, T$. Then
the MLE of $\theta$ \eqref{eq:theta} is given by
\begin{align}
\hat{\theta} & =\sum_h\hat{\lambda}_{hT}.\label{eq:hat-theta}
\end{align}
which is unbiased due to additivity and Theorem \ref{thm:MLE}.

\section{Confidence Interval}
\label{sec:ci}
As described in Algorithm \ref{alg:data-model}, our model involves a single parameter of interest $\theta$ but a large set of nuisance
parameters $\lambda_{ht}$ and $\pi_{ht}$. Due to the complexity
of the model, it is hard to derive an ``exact'' confidence interval
for $\theta$. 
As such, we consider several possible strategies for computing confidence intervals (CIs), then assess the relative metris of each strategy.

We first describe two standard approaches
based on a parametric bootstrap and a normal approximation (Wald CI) respectively.
These methods are expected to work well asymptotically as the expected True Positive event count increases, i.e., as the mileage $m$ goes to infinity. However, the CIs may be less reliable when the true positive events are rare. Due to the product importance of these rare event settings, we propose a
novel approach by approximating the exact model with a simpler model
-- weighted sum of independent Poissons (WSIP) -- and then applying
the Gamma method (\cite{fay1997confidence}), one of the
most established methods for WSIP. The Gamma method is expected to
work well not only for scenarios with high event counts (common events), but also for scenarios
with relatively few events (rare events).

\subsection{Parametric Bootstrap CI}

With the mileage $m$ fixed, conceptually, the model for the data
generation can be written as $p(\mathcal{D}|\{m,H,T,\{\lambda_{ht}, \pi_{ht}\}\})$,
where $\mathcal{D}$ consists of $\{e_{ht}:1\leq h\leq H,0\leq t\leq T\}$
and $\{n_{ht}:1\leq h\leq H,1\leq t\leq T\}$. The detailed data generation
procedure is provided by Algorithm \ref{alg:data-model}. This is
a parametric model, and the parametric bootstrap CI \cite{efron1992introduction}
can be constructed as follows:

Let $\{\hat{\lambda}_{ht},\hat{\pi}_{ht}\}$ be the point estimates
of $\{\lambda_{ht},\pi_{ht}\}$ respectively, where $\hat{\lambda}_{ht}$s
are obtained by applying the MLE estimator described in Theorem \ref{thm:MLE} to each stratum $h$ separately, and the following empirical estimate is used for $\pi_{ht}$ according to \eqref{eq:sample}:
\begin{align}
\hat{\pi}_{ht} & =\frac{n_{ht}}{e_{h(t-1)}}\label{eq:pi-hat}
\end{align}
which is also MLE when $n_{ht}>1$. It may be worth noting
that when $n_{ht}=1$ and $e_{h(t-1)}>1$, the MLE for $\pi_{ht}$ by the sampling model \eqref{eq:sample} is simply 0, which
is close to \eqref{eq:pi-hat} for large $e_{h(t-1)}$ but would be
less desirable in general.

Then draw $B$ i.i.d. samples of $\mathcal{D}$ according
to the generative process outlined in Algorithm \ref{alg:data-model} by plugging in the parameter estimates from the observed data. For $b=1,\cdots,B$,
\begin{align*}
\mathcal{D}^{(b)} & \sim p(\cdot|\{m,H,T,\{\hat{\lambda}_{ht}, \hat{\pi}_{ht}\}\}).
\end{align*}
For each draw, calculate the parameter estimates $\{\hat{\lambda}_{ht}^{(b)},\hat{\pi}_{ht}^{(b)}\}$
by replacing $\{e_{ht},n_{ht}\}$ in \eqref{eq:lambda-hat} and \eqref{eq:pi-hat}
with corresponding quantities in $\mathcal{D}^{(b)}$, and let
\begin{align*}
\hat{\theta}^{(b)} & =\sum_{h=1}^{H}\hat{\lambda}_{hT}^{(b)}.
\end{align*}

The $1-\alpha$ confidence interval for $\theta$  is given by the $\alpha/2$ and $1-\alpha/2$ quantiles
of $\{\hat{\theta}^{(b)}:b=1,\cdots,B\}$.

\subsection{Wald CI}

In order to construct the Wald CI, we need to estimate $var(\hat{\theta})$.
Unfortunately, the exact formula for $var(\hat{\theta})$ is extremely
complicated. Instead, we derive a first order approximation by assuming
$m$ is large. 

It may be interesting to note that $\hat{\lambda}_{hT}$ as defined by \eqref{eq:hat-lambda-hT} with index $h$ added to indicate the stratum can be rewritten as
\begin{align*}
\hat{\lambda}_{hT} & =\frac{e_{hT}}{m\hat{\pi}_{h}}\\
\end{align*}
where
\begin{align*}
\hat{\pi}_{h} & =\prod_{t=1}^{T}\hat{\pi}_{ht}
\end{align*}
and $\hat{\pi}_{ht}$s are defined by \eqref{eq:pi-hat}. Thus $\hat{\theta}$
in \eqref{eq:hat-theta} can be rewritten as
\begin{align}
\hat{\theta} & =\frac{1}{m}\sum_{h=1}^{H}\frac{e_{hT}}{\hat{\pi}_{h}}. \label{eq:theta_hat_pihat}
\end{align}

\begin{lem}
\label{lem:poisson} If a random variable $X\sim Poisson(\lambda)$,
and $Y\mid X\sim Binomial(X, \pi)$, then $Y\sim Poisson(\lambda\pi)$. 
\end{lem}

For large $m$, we may approximate the model \eqref{eq:sample} with subscript $h$ added to indicate stratum as
\begin{align*}
n_{ht} & \sim Binomial(e_{h(t-1)},\pi_{ht}).
\end{align*}
Then by iteratively applying Lemma \ref{lem:poisson}, we have 
\begin{align}
e_{hT} & \sim Poisson(m\lambda_{hT}\pi_{h}),\label{eq:e_hT-is-Poisson}
\end{align}
where $\pi_h = \prod_{t=1}^T\pi_{ht}$.
By combining this with \eqref{eq:theta_hat_pihat}, we can obtain the following result. 
\begin{thm}
\label{thm:poisson-normal}Assume $\pi_{ht}>0$ and $\lambda_{ht}>0$
for all $h, t$. As $m\rightarrow\infty$, we have
\begin{align*}
\sqrt{m}(\hat{\theta}-\theta) & \rightarrow\mathcal{N}(0,\sum_{h=1}^{H}\lambda_{hT}\pi_{h}^{-1}),
\end{align*}
where $\mathcal{N}(\mu,\sigma^2)$ denotes the normal distribution with mean $\mu$ and variance $\sigma^2$.
\end{thm}

As $m\rightarrow\infty,$ $\hat{\pi}_{ht}\rightarrow\pi_{ht}$ and
thus $\hat{\pi}_{h}\rightarrow\pi_{h}$, the proof follows from 
Slutsky's theorem.

From Theorem \ref{thm:poisson-normal}, the Wald CI can then be constructed
by $\hat{\theta}\pm z_{\alpha/2}\sqrt{\frac{1}{m}\sum_{h=1}^{H}\hat{\lambda}_{hT}\hat{\pi}_{h}^{-1}}$
based on the plug-in principle, where $z_{\alpha/2}$ is the $1-\alpha/2$
quantile of $\mathcal{N}(0,1)$.

As is well known, proper coverage from the Wald CI requires the asymptotic normality to
hold well. As such, it may not be desirable when the expected event counts are low (i.e., for rare events or for low mileage); for example, its lower bound may become negative. 

\subsection{Gamma CI based on an approximate weighted sum of independent Poissons}
Given the limitations of the normal asymptotic approximation outlined above, we propose an alternative approach to estimating the CIs for $\hat{\theta}$. Consider the estimate in \eqref{eq:theta_hat_pihat}, which can be approximated by the quantity below: 
\begin{align*}
\hat{\theta} & \approx\frac{1}{m}\sum_{h=1}^{H}\frac{e_{hT}}{\pi_{h}},
\end{align*}
where by \eqref{eq:e_hT-is-Poisson} the right hand side is a weighted sum
of independent Poissons (WSIP). The approximate WSIP becomes clear by rewriting $\hat{\theta}$ as,
\begin{align*}
\hat{\theta} & =\sum_{h=1}^{H}w_{h}e_{hT}\\
w_{h} & =\frac{1}{m\hat{\pi}_{h}}\text{ for }h=1,\cdots,H.
\end{align*}

The distribution for WSIP random variables has been studied extensively
in the statistics literature \cite{ng2008confidence,swift2010simulation,tiwari2006efficient,fay2017confidence}.
A WSIP is close to normally distributed when there are high event counts, but can be far from normal for rare events. This again motivates the consideration of a CI methodology that does not rely on normality. 
In particular, the original Gamma method proposed by \cite{fay1997confidence}
has been conjectured and also shown numerically to provide good coverage properties for the distribution mean
in all simulated scenarios \cite{ng2008confidence,fay2017confidence}.

Our third CI methodology leverages these WSIP properties from the literature and can be described as follows:
\begin{itemize}
\item The lower bound is the $\alpha/2$ quantile of the Gamma distribution
with mean $\hat{\theta}$ and variance $\sum_{h=1}^{H}w_{h}^{2}e_{hT}$. 
\item The upper bound is the $1-\alpha/2$ quantile of the Gamma distribution
with mean $\hat{\theta}+w_{M}$ and variance $\sum_{h=1}^{H}w_{h}^{2}e_{hT}+w_{M}^2$,
where $w_{M}=\max_{1\leq h\leq H}w_{h}$. 
\end{itemize}
Note that the above is based on the original Gamma method, but one may also apply other Gamma methods for WSIP such as the mid-p Gamma interval \cite{fay2017confidence}. 

\section{Numerical Studies}
\label{sec:sim}
\subsection{Simulation Setup}
We provide several simulation studies to assess and compare the performance of the three confidence interval methods defined in Section \ref{sec:ci}. The first two simulation studies compare the performance when the behavioral event is rare vs. when the behavioral event is common, using fixed values for most parameters. The third study illustrates that these coverage properties hold for a wide range of parameter values, beyond the specific values chosen for the first two examples. Throughout the simulation section we set $m = 1$ for convenience (the unit could be 1 mile, 10K miles, or one million miles).

In each study, we simulate the data in a consistent manner. Following Algorithm \ref{alg:data-model}, we generate the $x_{hs} \sim Poisson(\lambda_{hs}), h = 1, \hdots, H$ and $s = 0, \dots, T$. We set $H=5$ strata and assume $T=3$ tiers. Then $e_{h0} = \sum_{s=0}^T x_{hs}$ will be the total number of candidate events entering the tiered event review pipeline in stratum $h$. 

For the first two simulation studies in Section \ref{subsec: sim_common} and \ref{subsec: sim_rare}, we set the sampling rates
% $$
% \{ \pi_{ht} \} = 
% \begin{pmatrix}
% \pi_1 & 0.8 & 1 \\
% \pi_1 & 0.8 & 1 \\
% \pi_1 & 1 & 1 \\
% \pi_1 & 0.8 & 1 \\
% \pi_1 & 1 & 1
% \end{pmatrix}.
% $$
%
%
$$
\{ \pi_{ht} \} = 
\begin{pmatrix}
\pi_1 & 0.5 & 0.95 \\
\pi_1 & 0.6 & 0.96 \\
\pi_1 & 0.7 & 0.97 \\
\pi_1 & 0.8 & 0.98 \\
\pi_1 & 0.9 & 0.99
\end{pmatrix}.
$$

Note that when this event review methodology is put into practice one can imagine a fixed set of events to review, and the Tier 1 sampling rate will steadily increase over time as reviewers complete more events. We simulate this event review progression by varying the sampling rate at Tier 1, $\pi_1\in [0.1, 1.0]$. In this way, we can verify whether the confidence intervals have good coverage properties throughout the entire event review process.

For the third simulation study in Section \ref{subsec: sim_comprehensive}, we %put in more flexibility for the sampling rates. 
consider a wide range of scenarios.
For each scenario, we generate each Tier 1 sampling rate $\pi_{h1} \sim \text{Uniform }(0,1)$, and generate $\pi_{ht} \sim \text{Uniform } (\pi_{h(t-1)},1), t = 2,3,$ so that the sampling rate for each strata will not decrease as Tier goes up. $\{\lambda_{ht}\}$ is specified later.

Given each set of parameter values ($\{\lambda_{ht}\}$ and $\{\pi_{ht}\}$), we simulate the tiered partial event review data $\{e_{ht}\}$ and $\{n_{ht}\}$ following Algorithm \ref{alg:data-model} and replicate 1000 times. This allows us to obtain the empirical coverage and average CI width under each set of parameters.

\begin{figure}
\includegraphics[width=0.8\textwidth]{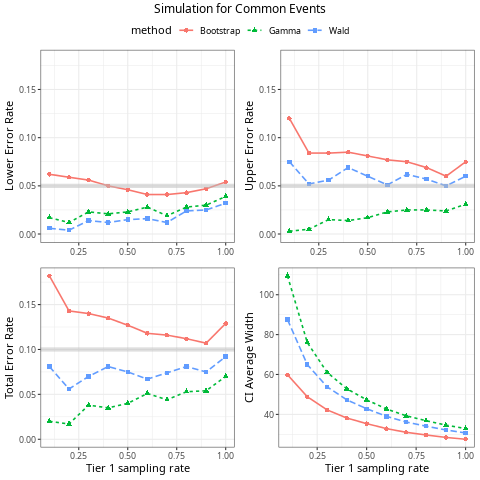}
\centering
\caption{\label{fig:dense} Common Events. Average one-sided coverage error of $90\%$ confidence intervals for lower bound (top left panel), upper bound (top right panel), total coverage error (bottom left panel) and average lengths of $90\%$ confidence intervals (bottom right panel). These results are based on estimating the total positive event rate $\theta=58$ under model \eqref{eq:model} over 1000 replications. The horizontal lines in the first three panels represent the test level $\alpha=0.05$.}
\end{figure}

\subsection{Simulation for Common Events}
\label{subsec: sim_common}
Some behavioral events are relatively common, and the observed event counts will be relatively high. As described earlier, this can result in better behaved quantities that more closely follow a normal distribution. Here we assess the performance of our proposed CI methodologies in this setting. 

To represent this common event scenario, the values below are used for $\{\lambda_{ht}\}$: 
$$\{\lambda_{ht}\} = 
\begin{pmatrix}
10 & 5 & 2.5 & 18 \\
20 & 15 & 25 & 10 \\
20 & 30 & 8 & 5 \\
5 & 6 & 25 & 10 \\
30 & 12 & 4 & 15 
\end{pmatrix}.$$
With this parameter setting, $\theta = 58$. 

Figure \ref{fig:dense} illustrates the coverage properties for the three CI methods when the behavioral event is common. The x-axis represents the sampling rate, allowing us to assess how well these CIs fulfill the product requirement of adequate coverage regardless of the sampling rate. For all three methods, the coverage rates are more biased (over- or under-covering) when sampling rates are low, with improvements in coverage as event review progresses. As sampling rates increase, the parametric bootstrap rapidly converges close to expected coverage rates, and also produces the narrowest intervals. However, the under-coverage at low sampling rates is non-negligible (82\% coverage vs. the desired 90\%).  In contrast, the Gamma CI consistently over-covers (98\% coverage for low sampling rates) and produces the widest intervals. However, the Gamma method also provides the most symmetric errors, with similar error rates for the upper and lower bounds; this guarantees that the intervals from this approach will not systematically under- or over-estimate the rates. Finally, the Wald CI falls between the other two methods, with intervals that somewhat over-cover.

\begin{figure}
\includegraphics[width=0.8\textwidth]{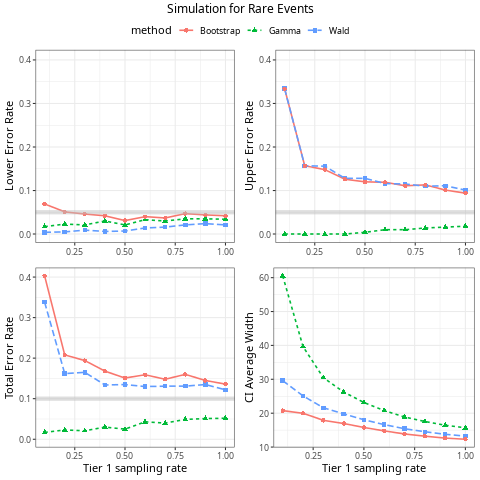}
\centering
\caption{\label{fig:sparse} 
Rare Events. Average one-sided coverage error of $90\%$ confidence intervals for lower bound (top left panel), upper bound (top right panel), total coverage error (bottom left panel) and average lengths of $90\%$ confidence intervals (bottom right panel). These results are based on estimating the total positive event rate $\theta=11$ under model \eqref{eq:model} over 1000 replications. The horizontal lines in the first three panels represent the test level $\alpha=0.05$.}
\end{figure}

\subsection{Simulation for Rare Events}
\label{subsec: sim_rare}
Next, we consider more rare events, where the observed event counts will be quite low. In this setting, the estimates are not going to cleanly follow a normal distribution, making the CI calculation potentially more challenging. Here we assess the performance of our proposed CI methodologies in this setting, again tracking how the coverage properties evolve as more events receive review. 

To represent this scenario, we use the same $\lambda_{ht}$ as before for $t=0,1,2$ but reduce $\lambda_{h3}$ for all $h$:
%$\overrightarrow{\lambda_{\cdot 0}}$ as before, then set the underlying Poisson rates for the following tiers as follows:
$$\{\lambda_{ht}\} = 
\begin{pmatrix}
10 & 5 & 2.5 & 4 \\
20 & 15 & 25 & 2 \\
20 & 30 & 8 & 1 \\
5 & 6 & 25 & 2 \\
30 & 12 & 4 & 2 
\end{pmatrix}.$$
%for $h = 1, \hdots, 5$ and $t = 1, \hdots, 3$. 
With this parameter setting, $\theta = 11$.
%\textbf{RACHEL - Please add the implied Theta value - thank you!}, which corresponds to the rate of a more common event.

%Figure \ref{fig:sparse} reports the performance comparison between the three methods similar to the previous study. The patterns in this figure are largely similar to that in Figure \ref{fig:dense}, except that for the upper bound, both Wald method and parametric bootstrap can be severely under-covered, with the non-coverage error rate almost 100\% above the expected value, while the Gamma method again has the right coverage, but it can be quite conservative, and has the widest CI on average. For safety-critical metrics, the Gamma method should be preferred.

Figure \ref{fig:sparse} reports the coverage properties for the three confidence interval methods when the behavioral event is rare. The patterns in this figure are largely similar to the patterns in Figure \ref{fig:dense}. The primary difference between the rare event case and the common event case is the performance at low sampling rates. The under-coverage by the Wald and Parametric Bootstrap is severe at the lowest sampling rates, off by as much as 30 percentage points (60\% coverage vs. the desired 90\%).  In this case, the over-coverage produced by the Gamma CI (97\% vs. the desired 90\%) is clearly preferable. That said, the Gamma CI is much wider than the Parametric Bootstrap, as much as 3x wider for low sampling rates.

\begin{figure}[!t]
\includegraphics[width=0.8\textwidth]{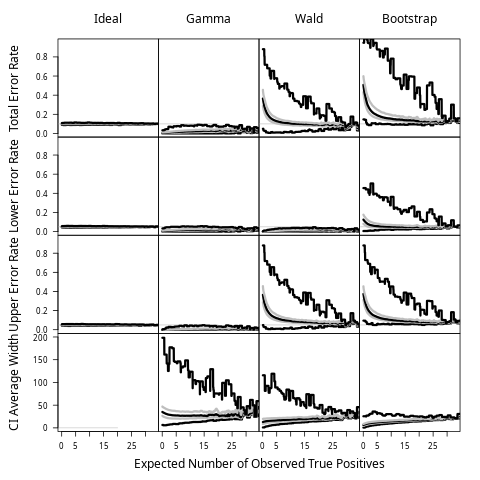}
\centering
\caption{\label{fig:comprehensive} Average error rates and CI widths of 90$\%$ confidence intervals by different methods versus the the expected number of observed True Positives from $m=1$ mileage over $10^5$ scenarios. The lines are the minimum (black), 25$\%$ percentile (gray), median (black), 75$\%$ percentile (gray), and maximum (black) calculated from moving subsets of the data, where the subset
for the value at expected number of observed True Positives includes all observations with var(weights) $\in$ [expected number - 1, expected number + 1]. Reference lines for the error rate panels are drawn at the nominal error rates for a central 90$\%$ confidence interval.}
\end{figure}

\subsection{Comprehensive simulation}
\label{subsec: sim_comprehensive}
In each of the simulations above, the majority of parameters were fixed, while the first tier sampling rate was allowed to vary. In our final simulation we continue to vary the sampling rate, as described earlier, but consider $10^5$ different scenarios covering a wide range of parameter values $\lambda_{ht}$.
In each scenario, we simulate the underlying Poisson rates as follows: $$\lambda_{ht} \sim \text{Exp}(\mu_{ht}), h=1,\cdots,H\text{ and }t=0,\cdots T,$$ where $\mu_{ht} \sim \text{Uniform}(1, 4)$ independently. This simulation study is inspired by the simulation setup in \cite{fay2017confidence}.

Figure \ref{fig:comprehensive} reports the performance comparison. Unlike Figures \ref{fig:dense} and \ref{fig:sparse} where the x-axis represented a range of sampling rates, the x-axis here corresponds to the expected observed number of true positive events per mile, i.e. $\frac{1}{m}\sum_h E(e_{hT})$ which equals $\sum_h \lambda_{hT} \prod_{t=1}^T\pi_{ht}$. In the earlier examples, the underlying event rates were fixed and only the sampling rates varied; in this study, both the sampling rates and the event rates are varying simultaneously. In this setup, this x-axis allows for consideration of how the coverage properties evolve in the asymptotic sense, regardless of whether it is increased review or increased event occurrence. 

The pattern is overall similar to the previous examples, for example, for the lower bound, both Gamma method and Wald method maintain the right coverage. However, the Wald method can be badly under-covered for its upper bound, while parametric bootstrap can be badly under-covered for both lower bound and upper bound. On the other hand, the Gamma method is the only method that always produces reliable confidence intervals for all scenarios. Not surprisingly, this method also produces the widest intervals on average.

\section{Summary and Discussion}\label{sec:discussion}
Before launching any new autonomous driving software onto public roads, it is critical to first evaluate the driving quality of the software. To this end, estimation of behavioral event rates has become a statistical focus area within the autonomous vehicles industry. This estimation process often requires more sophisticated statistical consideration than one might expect at first glance.

As described earlier, accurate estimation of these event rates relies on accurate detection of the behavioral events within the driving logs. For many such events, we rely on algorithmic detectors combined with a partial tiered event review process where human reviewers reject or escalate the candidate events at each tier. Importantly, this event review process is a streaming process that rarely provides complete review of all events in the system. As a result, there are many candidate events with missing outcome labels, complicating the total event count estimation. In addition, each of these events may have already progressed to a different stage within the event review pipeline, and this partial information can be used to better estimate the candidate event's likelihood of being a True Positive.

While the point estimate for the event rate is fairly intuitive, how to construct a reliable confidence interval for the rate is less straightforward. As defined, the data generation model has a very large number of latent variables and nuisance parameters. In addition, the practical application of this work within the autonomous vehicles industry necessitates strong performance for a range of behavioral event types ranging from benign common events to more safety-relevant rare events. As shown through the simulation studies, the different CI methods have different strengths and weaknesses, and these vary depending on the sampling rate and/or event frequency. Thus, the choice of confidence interval may depend on the characteristics of the behavioral event of interest. For a blanket approach that ensures consistent conservative coverage, we propose the Gamma method based on the weighted sum of independent Poissons.  

The Parametric Bootstrap CI is consistently the narrowest CI, and it also frequently under-covers. For common behavioral events with a reasonable rate of review, this method's under-coverage is typically not severe. However, when the sampling rate is low, e.g. under 25\%, this method can drastically undercover. 

The Wald CI is largely similar to the Parametric Bootstrap, except that the Wald CIs are somewhat wider and provide greater coverage. For  applications where it is practical to require minimum rates of review, this method may be attractive to some users. With sufficient guardrails (i.e., only for common events and required minimum rates of review before reporting), this method provides close to the intended coverage and provides narrower intervals that can inform downstream product decision-making by stakeholders. Of course, we do \emph{not} recommend use of this method unless the practitioner can guarantee the application is a setting where the method achieves desired coverage.

In many cases the above methods are not appropriate, for example when the user does not have \emph{a priori} information about the relative frequency of the behavioral event, the behavioral event is known to be rare, or there is a product requirement to report rates at \emph{all} stages of the event review process. In each of these settings, we recommend using the Gamma method based on the weighted sum of independent Poissons, which was shown to maintain consistent coverage across all simulations. 

The Gamma method is a conservative strategy, which consistently over-covers relative to the target coverage rate and has interval widths that are consistently wider than the other methods. While this over-coverage is not strictly desirable (we'd prefer to achieve the target coverage rate precisely!), we believe that a conservative approach is the best strategy when estimating the rate of high priority safety-relevant behavioral events. In addition, the upper and lower error rates are generally quite similar, giving confidence that these intervals will not systematically over- or under-estimate the rate of occurrence for undesirable behavioral events.
These properties are particularly important given the very-real implications of launching ADS on public roads. 

In summary, we have outlined the benefits of the streaming partial tiered event review, but we have also illustrated that estimating confidence intervals for behavioral event rates under this scheme is quite challenging. While the Gamma method had the most consistent coverage, the utility of the resulting confidence intervals is limited by their excessive width and over-coverage in many settings. These findings are similar to the properties observed in the original Gamma method \cite{fay1997confidence}. Future work should further refine and improve these strategies while maintaining the strong properties of the existing approach. 

A vast literature exists on the subject of sampling \cite{cochran1977sampling},
but application of the complex mixed sampling technique as presented
is new for the area of autonomous vehicles. How to construct a valid
confidence interval in the presence of many nuisance parameters is still very much an open problem for rare events, and the standard approaches such as bootstrap and Wald
CI may fail easily for small event counts (e.g. \cite{wang2013note}). The
most relevant work to ours include the profile likelihood-based method
\cite{venzon1988method}, the Buehler method \cite{buehler1957confidence}, and the hybrid sampling method \cite{chuang2000hybrid},
which however cannot be easily applied for our problem, see \cite{sen2009unified,rolke2005limits} for some recent reviews on the subject.

\section{Appendix: proof of Theorem \ref{thm:MLE}}

We first present a theorem about the relationship between
Poisson, multivariate Hypergeometric and multinomial distributions,
which may be of independent interest and is used to prove Theorem \ref{thm:MLE}.
We then proceed to derive the
MLE by making use of the EM algorithm. Results 1) and 3) of Theorem \ref{thm:MLE}
follows directly from Proposition \ref{prop:EM-solution}, and result
2) follows from Lemma \ref{lem:unbiasedness}. In the end, we also
provide an illustrative example for $T=2$ on how to solve the EM
equations.

\subsection{Relation between Poisson, multivariate Hypergeometric and multinomial distributions}

Consider an urn of balls, where each ball has a color $k$ with $k\in\{0,1,\cdots,K\}$.
Let $z_{k}$ be the number of balls with color $k$, and let $s_{z}=\sum_{k=0}^{K}z_{k}$
and $\vec{z}$ be the column vector of $(z_{0},\cdots,z_{K})$.

Let $n\leq s_{z}$ be a random number whose distribution only depends
on $s_{z}$. 

Randomly draw $n\geq 1$ balls from the urn without replacement, and let
$y_{k}$ be the number of balls with color $k$ for $k\in\{1,\cdots,K\}$.
Let $s_{y}=\sum_{k=1}^{K}y_{k}$ and $\vec{y}$ be the column vector of $(y_1, \cdots, y_K)$.

Then, conditional on $(\vec{z},n)$, $\vec{y}$ follows a multivariate
Hypergeometric distribution:
\begin{align*}
p(\vec{y}|\vec{z},n) & =\frac{{z_{0} \choose n-s_{y}}\prod_{k=1}^{K}{z_{k} \choose y_{k}}}{{s_{z} \choose n}}.
\end{align*}

\begin{thm}
\label{lem:poisson-hypergeometric-multinomial}If $z_k \sim Poisson(\lambda_k)$, $k=0,1,\cdots,K$, then the following results
hold.

1) Conditional on $(\vec{y},s_{z},n)$, the distribution of $\vec{z}$
is Multinomial with shifted mean:
\begin{align*}
\vec{z}|(\vec{y},s_{z},n) & \sim\left[\begin{array}{c}
n-s_{y}\\
\vec{y}
\end{array}\right]+Multinomial(s_{z}-n,\left[\begin{array}{c}
\Lambda^{-1}\lambda_{0}\\
\cdots\\
\Lambda^{-1}\lambda_{K}
\end{array}\right])
\end{align*}
where $\Lambda=\sum_{k=0}^{K}\lambda_{k}$.

2) Conditional on $s_{y}$, $\vec{y}$ is independent of $(s_{z},n)$,
and follows a conditional Poisson distribution:
\begin{align*}
p(\vec{y}|s_{y}) & =C(s_{y})\times\prod_{k=1}^{K}\frac{e^{-\lambda_{k}}\lambda_{k}^{y_{k}}}{y_{k}!}
\end{align*}
where $C(s)=\left(\sum_{y_{1}+\cdots+y_{k}=s}\prod_{k=1}^{K}\frac{e^{-\lambda_{k}}\lambda_{k}^{y_{k}}}{y_{k}!}\right){}^{-1}$
is the normalization factor and only depends on $s$. In other words,
\begin{align*}
\vec{y}|(s_{z},n,s_{y}) & \sim\vec{u}|\sum_{k=1}^{K}u_{k}=s_{y}
\end{align*}
where $\vec{u}=(u_{1},\cdots,u_{K})$ and $u_{k}\sim Poisson(\lambda_{k})$ are independent. 
\end{thm}

\begin{proof}
Note that the joint probability can be written as
\begin{align*}
p(\vec{z},n,\vec{y}) & =p(\vec{z})p(n|s_{z})p(\vec{y}|\vec{z},n).
\end{align*}

So given $(\vec{y},s_{z},n)$, we have
\begin{align*}
p(\vec{z}|\vec{y},s_{z},n) & \propto p(\vec{z})p(\vec{y}|\vec{z},n)\\
 & =\prod_{k=0}^{K}\frac{e^{-\lambda_{k}}\lambda_{k}^{z_{k}}}{z_{k}!}\times\frac{{z_{0} \choose n-s_{y}}\prod_{k=1}^{K}{z_{k} \choose y_{k}}}{{s_{z} \choose n}}\\
 & \propto\frac{\lambda_{0}^{z_{0}}\prod_{k=1}^{K}\lambda_{k}^{z_{k}-y_{k}}}{(z_{0}-n+s_{y})!\prod_{k=1}^{K}(z_{k}-y_{k})!}\\
 & \propto{s_{z}-n \choose \begin{array}{c}
\begin{array}{cccc}
(z_{0}-(n-s_{y})) & (z_{1}-y_{1}) & \cdots & (z_{K}-y_{k})\end{array}\end{array}}p_{0}^{z_{0}-(n-s_{y})}\prod_{k=1}^{K}p_{k}^{z_{k}-y_{k}}
\end{align*}
where $p_k = \Lambda^{-1}\lambda_k$. This completes the proof of 1).

The proof of 2) follows from the Bayes theorem, i.e. 
\begin{align*}
p(\vec{y}|s_{z},n,s_{y}) & =\frac{p(\vec{z},\vec{y}|s_{z},n,s_{y})}{p(\vec{z}|\vec{y},s_{z},n)}\\
 & \propto\frac{p(\vec{z})p(\vec{y}|\vec{z},n)}{p(\vec{z}|\vec{y},s_{z},n)}.
\end{align*}
\end{proof}

\subsection{Likelihood inference for a single stratum}

%Since the random variates in different strata are independent, it is sufficient to 
%show that $\hat{\lambda}_{hT}$ %defined by (\ref{eq:to-be-removed})
%is the maximum likelihood estimate of $\lambda_{hT}$ for each stratum
%$h$.

%For notational simplicity, we drop $h$ hereafter, that is, 
The data for a single stratum can be described as $(e_{0},n_{1},e_{1},\cdots,n_{T},e_{T})$.
Note that in case the partial event review process is terminated earlier due to no escalation, say $e_t=0$ for some tier $t<T$, then there will be no data from higher tiers (i.e., $n_s=0$ and $e_s=0$ for $s>t$ according to Algorithm \ref{alg:data-model-one-stratum}), so we may simply reset $T$ to $t$ for the data likelihood calculation.
Let $\vec{e}=(e_{0},\cdots,e_{T})$ and $\vec{n}=(n_{1},\cdots,n_{T})$.
The observed data consists of $\vec{e}$ and $\vec{n}$. Without loss
of generality, we will take the mileage $m=1$.

To follow Algorithm \ref{alg:data-model-one-stratum}, let $\mathcal{E}_0$ consist of all $e_0$ candidate events. Let $\mathcal{E}_t$ be the 
set of $e_t$ events which are escalated by Tier-$t$ for $t=1, \cdots, T$, so $\mathcal{E}_T$ consists of the $e_T$ events which are True Positive. Furthermore, each candidate event is associated with a label,
which is either FP-$t$ (if the event would be escalated by Tier-$(t-1)$ but rejected as False Positive by Tier-$t$) for some $t\in \{1,\cdots, T\}$,
or TP (if the event would not be rejected by any Tier), under the complete event review configuration.

For $s\in \{0, \cdots, T-1\}$, let $x_{ts}$ be the number of FP-$t$ events in $\mathcal{E}_s$, $\forall s\leq t < T$, and let $x_{Ts}$ be the number of TP events in $\mathcal{E}_s$, then 
$$e_s = \sum_{t=s}^Tx_{ts}.$$ 
For convenience, we also introduce $x_{TT}\equiv e_T$. With some abuse of notation, $(x_{00},\cdots,x_{T0})$ is identical to $(x_0, \cdots,x_T)$ as defined by \eqref{eq:model}.

Let $\vec{x}_{s}=(x_{ss},\cdots,x_{Ts})$, for $s=0,\cdots,T$. Table \ref{tab:Table-T} summarizes the
notation of $\{x_{ts}:0\leq s\leq t\leq T\}$ and relation with $\vec{e}$.

\begin{table}
\caption{\label{tab:Table-T}Notation for $\{x_{ts}:0\leq s\leq t\leq T\}$, where $x_{ts}$ is the number of FP-$t$ events escalated by Tier-$s$ for $1\leq s\leq t < T$, $x_{Ts}$ is the number of TP events escalated by Tier-$s$, and $\{x_{t0}\}$ is identical to $\{x_t\}$ as defined in Algorithm \ref{alg:data-model-one-stratum}.}
\centering{}%
\begin{tabular}{|c||c|c|c|c|c|}
\hline 
 & Tier-0 & Tier-1 & $\cdots$ & Tier-$(T-1)$ & Tier-$T$ \tabularnewline
\hline 
\hline 
FP-1  & $x_{00}$  & 0  & 0  & 0  & 0 \tabularnewline
\hline 
FP-2  & $x_{10}$ & $x_{11}$  & 0  & 0  & 0 \tabularnewline
\hline 
$\cdots$  & $\cdots$  & $\cdots$ & $\cdots$  & 0 & 0 \tabularnewline
\hline 
FP-$T$  & $x_{(T-1)0}$ & $x_{(T-1)1}$ & $\cdots$ & $x_{(T-1)(T-1)}$  & 0 \tabularnewline
\hline 
TP & $x_{T0}$ & $x_{T1}$ & $\cdots$ & $x_{T(T-1)}$ & $x_{TT}$ \tabularnewline
\hline \hline
Total & $e_0$ & $e_1$ & $\cdots$ & $e_{T-1}$ & $e_T$ \tabularnewline
\hline
\end{tabular}
\end{table}

According to the sampling procedure described in Algorithm \ref{alg:data-model-one-stratum},
the following conditional independence holds:
\begin{align}
P(\vec{x}_{s}|\vec{x}_{0},\cdots,\vec{x}_{s-1},n_{1},\cdots,n_{s}) & =P(\vec{x}_{s}|\vec{x}_{s-1},n_{s}).\label{eq:conditional-independence-sampling}
\end{align}
Also note that $n_{s}$ only depends on $e_{s-1}$. Therefore the
full likelihood function of $\{\vec{x}_s: 0\leq s\leq T\}$ and $\{n_s: 1\leq s\leq T\}$ can be written as 
\begin{align}
L & =P(\vec{x}_{0})\times\prod_{s=1}^{T-1}P(\vec{x}_{s}|\vec{x}_{s-1},n_{s})\times P(e_{T}|\vec{x}_{T-1},n_{T})\times\prod_{s=1}^{T}P(n_{s}|e_{s-1}).\label{eq:full-likelihood}
\end{align}

According to the Poisson assumption described in Algorithm \ref{alg:data-model-one-stratum},
\begin{align*}
P(\vec{x}_{0}) & =\prod_{t=0}^{T}\frac{e^{-\lambda_{t}}\lambda_{t}^{x_{t0}}}{x_{t0}!}
\end{align*}
According to the sampling procedure, 
\begin{align}
P(\vec{x}_{s}|\vec{x}_{s-1},n_{s}) & =\frac{{x_{(s-1)(s-1)} \choose n_{s}-e_{s}}\prod_{t=s}^{T}{x_{t(s-1)} \choose x_{ts}}}{{e_{s-1} \choose n_{s}}}.\label{eq:sampling-dependence}
\end{align}

\begin{lem}
\label{lem:unbiasedness}Let $\hat{\Lambda}_t$ be defined by \eqref{eq:Lambda-hat}, then $\forall t\in \{0, \cdots, T\}$
\begin{align*}
E(\hat{\Lambda}_t) = \Lambda_t.
%E(\frac{\prod_{t=0}^{T}e_{t}}{\prod_{t=1}^{T}n_{t}}) & =\lambda_{T}.
\end{align*}
\end{lem}

\begin{proof}
It is sufficient to prove the result for $t=T$, since the proof is the same for other $t$.
By \eqref{eq:sampling-dependence}, conditional on $\vec{x}_{T-1}$
and $n_{T}$, $x_{TT}$ (which is $e_{T}$) follows a hypergeometric
distribution. Using \eqref{eq:conditional-independence-sampling},
for $e_{T-1}>0$, we have
\begin{align*}
E(e_{T}|\vec{x}_{0},\cdots,\vec{x}_{T-1},n_{1},\cdots,n_{T})= & E(e_{T}|\vec{x}_{T-1},n_{T})\\
= & \frac{x_{T(T-1)}n_{T}}{e_{T-1}}.
\end{align*}
Note that
\begin{align*}
    E(\hat{\Lambda}_T) & = E(\hat{\Lambda}_T I(e_{0:(T-1)} > 0)) \\
    & = E\left(\frac{\prod_{t=0}^{T}e_{t}}{\prod_{t=1}^{T}n_{t}} I(e_{0:(T-1)} > 0)\right),
\end{align*}
where $I(e_{0:t}>0)$ indicates whether $e_s>0$ for all $s\in \{0,\cdots, t\}$.
Then
\begin{align*}
E(\hat{\Lambda}_T) & =E\left(\frac{\prod_{t=0}^{T-1}e_{t}}{\prod_{t=1}^{T}n_{t}}I(e_{0:(T-1)} > 0)\times E(e_{T}|\vec{x}_{0},\cdots,\vec{x}_{T-1},n_{0},\cdots,n_{T})\right)\\
 & =E\left(\frac{\prod_{t=0}^{T-2}e_{t}}{\prod_{t=1}^{T-1}n_{t}}I(e_{0:(T-2)} > 0)\times x_{T(T-1)}\right).
\end{align*}
By iteratively applying \eqref{eq:conditional-independence-sampling} for $e_{s-1}>0$,
\begin{align*}
E(x_{Ts}|\vec{x}_{0},\cdots,\vec{x}_{s-1},n_{1},\cdots,n_{s}) & =E(x_{Ts}|\vec{x}_{s-1},n_{s})=\frac{x_{T(s-1)}n_{s}}{e_{s-1}}
\end{align*}
for $s=T-1,\cdots,1$, we get 
\begin{align*}
E(\hat{\Lambda}_T) & =E(x_{T0}I(e_0>0)) = E(x_{T0})
\end{align*}
which is $\lambda_{T}$ by \eqref{eq:model} since $x_{T0} \equiv x_T$.
\end{proof}
In order to drive the MLE, we apply the EM algorithm.

\subsection{The Expectation-Maximization algorithm}
\begin{lem}
The M-step simplifies to 
\begin{align}
\lambda_{t} & =E(x_{t0}|\vec{e},\vec{n})\text{ for }t=1,\cdots,T,\label{eq:M-equations}
\end{align}
which implies $\sum_{t=0}^{T}\lambda_{t}=e_{0}$, since $\sum_{t=0}^{T}x_{t0}\equiv e_{0}$. 
\end{lem}

\begin{proof}
Note that in the full likelihood function (\ref{eq:full-likelihood}),
only the term $P(\vec{x}_{0})$ involves the parameters of interest.
So the M-step, which maximizes $E(\log L|\vec{e},\vec{n})$ w.r.t.
$(\lambda_{0},\cdots,\lambda_{T})$, is equivalent to maximization
of 
\begin{align*}
\ell(\lambda_{0},\cdots,\lambda_{T}) & \equiv E(\log P(\vec{x}_{0})|\vec{e},\vec{n})\\
 & =\sum_{t=0}^{T}E(x_{t0}|\vec{e},\vec{n})\times\log\lambda_{t}-\lambda_{t}+const
\end{align*}
where $const$ does not depend on $(\lambda_{0},\cdots,\lambda_{T})$.

The conclusion follows from 
\begin{align*}
\frac{\partial\ell}{\partial\lambda_{t}} & =\frac{E(x_{t0}|\vec{e},\vec{n})}{\lambda_{t}}-1=0.
\end{align*}
\end{proof}
For $s=0,\cdots,T$ let 
\begin{align*}
\Lambda_{s} & \equiv\sum_{t=s}^{T}\lambda_{t}.
\end{align*}

In order to derive the E-step, we need the result below which follows
from Theorem \ref{lem:poisson-hypergeometric-multinomial}. The detailed proof
is omitted for conciseness.
\begin{lem}
\label{lem:EM-multinomial}Conditional on $(e_{0},\vec{x}_{1},n_{1})$,
$\vec{x}_{0}$ is independent of $(\vec{x}_{2},\cdots,\vec{x}_{T},n_{2},\cdots,n_{T})$,
and follows a multinomial distribution with mean shift as follows:
\begin{align*}
\vec{x}_{0}\Big|(e_{0},\vec{x}_{1},n_{1}) & \sim\left[\begin{array}{c}
n_{1}-e_{1}\\
\vec{x}_{1}
\end{array}\right]+Multinomial\left(e_{0}-n_{1},\left(\begin{array}{c}
\Lambda_{0}^{-1}\lambda_{0}\\
\cdots\\
\Lambda_{0}^{-1}\lambda_{T}
\end{array}\right)\right).
\end{align*}
Furthermore, for each $s\in\{1,\cdots,T-1\}$, conditional $(e_{s},\vec{x}_{s+1},n_{s+1})$,
$\vec{x}_{s}$ is independent of $(e_{0},\cdots,e_{s-1},\vec{x}_{s+2},\cdots,\vec{x}_{T},n_{1},\cdots,n_{s},n_{s+2},\cdots,n_{T})$,
and follows multinomial with mean shift:
\begin{align*}
\vec{x}_{s}\Big|(e_{s},\vec{x}_{s+1},n_{s+1}) & \sim\left[\begin{array}{c}
n_{s+1}-e_{s+1}\\
\vec{x}_{s+1}
\end{array}\right]+Multinomial\left(e_{s}-n_{s+1},\left(\begin{array}{c}
\Lambda_{s}^{-1}\lambda_{s}\\
\cdots\\
\Lambda_{s}^{-1}\lambda_{T}
\end{array}\right)\right).
\end{align*}
\end{lem}

The result below follows trivially from Lemma \ref{lem:EM-multinomial}. 
\begin{lem}
The E-step can be calculated as below: for $s=0,\cdots,T-1$,
\begin{align}
E(x_{ss}|\vec{e},\vec{n}) & =(n_{s+1}-e_{s+1})+(e_{s}-n_{s+1})\times\frac{\lambda_{s}}{\Lambda_{s}}\label{eq:E-eq-ss}\\
E(x_{ts}|\vec{e},\vec{n}) & =E(x_{t(s+1)}|\vec{e},\vec{n})+(e_{s}-n_{s+1})\times\frac{\lambda_{t}}{\Lambda_{s}},\text{for }t=s+1,\cdots,T.\label{eq:E-eq-st}
\end{align}
\end{lem}

The convergence point satisfies the equations from both the M-step
and E-step, i.e. \eqref{eq:M-equations}, \eqref{eq:E-eq-ss} and
\eqref{eq:E-eq-st}, which is summarized below. 
\begin{prop}
\label{prop:EM-solution}The EM algorithm converges to a unique solution
as below: 
\begin{align*}
\hat{\Lambda}_{t} & =\frac{\prod_{s=0}^{t}e_{s}}{\prod_{s=1}^{t}n_{s}}
\end{align*}
for $t=0,\cdots,T$. Thus the MLE of $(\lambda_{0},\cdots,\lambda_{T})$
is given by $\hat{\lambda}_{t}=\hat{\Lambda}_{t}-\hat{\Lambda}_{t+1}$
for $t=0,\cdots,T-1$ and $\hat{\lambda}_{T}=\hat{\Lambda}_{T}$. 
\end{prop}

The detail for solving the equations is omitted for conciseness. Instead,
we provide the detailed proof for $T=2$ to illustrate the idea.

\subsection{Illustration with $T=2$}

For illustration, below shows the complete E-step and M-step for $T=2$,
for which, the notation of $x_{st}$ is listed in Table \ref{tab:Table-T2}
for convenience.

\begin{table}
\caption{\label{tab:Table-T2}Notation for $\{x_{ts}:0\protect\leq s\protect\leq t\protect\leq2\}$.}
\centering{}%
% \begin{tabular}{|c||c|c|c||c|}
% \hline 
%  & FP-1  & FP-2  & TP  & Sum\tabularnewline
% \hline 
% \hline 
% Tier-0  & $x_{00}$  & $x_{01}$  & $x_{02}$  & $e_{0}$\tabularnewline
% \hline 
% Tier-1  & 0  & $x_{11}$  & $x_{12}$  & $e_{1}$\tabularnewline
% \hline 
% Tier-2  & 0 & 0 & $e_{2}$  & $e_{2}$\tabularnewline
% \hline 
% \end{tabular}
\begin{tabular}{|c||c|c|c|c|}
\hline 
 & Tier-0 & Tier-1 & Tier-2 \tabularnewline
\hline 
\hline 
FP-1  & $x_{00}$  & 0  & 0 \tabularnewline
\hline 
FP-2  & $x_{10}$ & $x_{11}$  & 0 \tabularnewline
\hline 
TP & $x_{20}$ & $x_{21}$ & $x_{22}$ \tabularnewline
\hline \hline
Total & $e_0$ & $e_1$ & $e_2$ \tabularnewline
\hline
\end{tabular}

\end{table}

According to the description of the EM algorithm in the previous section,
we have, with $obs\equiv(e_{0},n_{1},e_{1},n_{2},e_{2})$ , 
\begin{itemize}
\item M-step: 
\begin{align*}
\lambda_{t} & =E(x_{t0}|obs),\text{ for }t=0,1,2.
\end{align*}
\item E-step: 
\begin{align*}
E(x_{00}|obs) & =(e_{0}-n_{1})\frac{\lambda_{0}}{\sum_{t=0}^{2}\lambda_{t}}+(n_{1}-e_{1})\\
E(x_{10}|obs) & =(e_{0}-n_{1})\frac{\lambda_{1}}{\sum_{t=0}^{2}\lambda_{t}}+(e_{1}-n_{2})\frac{\lambda_{1}}{\sum_{t=1}^{2}\lambda_{t}}+(n_{2}-e_{2})\\
E(x_{20}|obs) & =(e_{0}-n_{1})\frac{\lambda_{2}}{\sum_{t=0}^{2}\lambda_{t}}+(e_{1}-n_{2})\frac{\lambda_{2}}{\sum_{t=1}^{2}\lambda_{t}}+e_{2}.
\end{align*}
\end{itemize}
The convergence point of the EM algorithm, i.e. $(\lambda_{0},\lambda_{1},\lambda_{2})$ satisfies the equations
in both E-step and M-step, which simplifies to: 
\begin{align}
\lambda_{0} & =(e_{0}-n_{1})\frac{\lambda_{0}}{e_{0}}+(n_{1}-e_{1})\label{eq:lambda0}\\
\lambda_{1} & =(e_{0}-n_{1})\frac{\lambda_{1}}{e_{0}}+(e_{1}-n_{2})\frac{\lambda_{1}}{\lambda_{1}+\lambda_{2}}+(n_{2}-e_{2})\label{eq:lambda1}\\
\lambda_{2} & =(e_{0}-n_{1})\frac{\lambda_{2}}{e_{0}}+(e_{1}-n_{2})\frac{\lambda_{2}}{\lambda_{1}+\lambda_{2}}+e_{2}.\label{eq:lambda2}
\end{align}

By summing up both sides of \eqref{eq:lambda0}, \eqref{eq:lambda1}
and \eqref{eq:lambda2}, we get 
\begin{align*}
\lambda_{0}+\lambda_{1}+\lambda_{2} & =e_{0}.
\end{align*}

By summing up both sides of \eqref{eq:lambda1} and \eqref{eq:lambda2},
we get 
\begin{align*}
\lambda_{1}+\lambda_{2} & =(e_{0}-n_{1})\frac{\lambda_{1}+\lambda_{2}}{e_{0}}+(e_{1}-n_{2})+(n_{2}-e_{2})+e_{2}
\end{align*}
and thus 
\begin{align*}
\lambda_{1}+\lambda_{2} & =\frac{e_{0}e_{1}}{n_{1}}.
\end{align*}

By plugging-in this back into \eqref{eq:lambda2}, we get 
\begin{align*}
\lambda_{2} & =\frac{e_{0}e_{1}e_{2}}{n_{1}n_{2}}
\end{align*}
which completes the proof.

\section*{Acknowledgment}
We would like to thank Henning Hohnhold and Kelvin Wu for many insightful discussions.

The event review processes described here are operationalized in partnership with multiple cross-functional teams. We would like to especially thank Vinit Arondekar, Jesse Breazeale, Han Hui, Jeffy Johns, Sidlak Malaki, Kentaro Takada, Ally Yang and Frank Zong for their collaboration.

\bibliographystyle{plain}
\bibliography{main}

\end{document}